\documentclass[12pt,draftcls,onecolumn]{IEEEtran}

\IEEEoverridecommandlockouts

\usepackage{algorithm,algorithmic,nicefrac}
\usepackage{amsmath,amsfonts,amssymb}
\usepackage{txfonts}
\usepackage{psfrag}
\usepackage{epsfig,graphicx,subfigure}
\usepackage{epstopdf}
\usepackage{stfloats}
\usepackage{epsf}

\usepackage{tikz}
\usetikzlibrary{shapes.misc}
\usetikzlibrary{decorations.pathreplacing}

\newtheorem{defn}{Definition}
\newtheorem{thm}{Theorem}
\newtheorem{lem}{Lemma}
\newtheorem{cor}{Corollary}
\newtheorem{cla}{Property}
\newtheorem{rem}{Remark}
\newtheorem{example}{Example}

\newcommand{\M}{\mathcal{M}}
\newcommand{\R}{\mathcal{R}}
\newcommand{\I}{\mathcal{I}}

\newcommand{\s}{\mathcal{S}}
\newcommand{\LL}{\mathcal{L}}

\title{\LARGE \bf
Distributed team formation in multi-agent systems:\\stability and approximation}

\author{Lorenzo Coviello and Massimo Franceschetti% <-this % stops a space
\thanks{This work was partially supported by the Army Research Office grant number W911NF-11-1-0363.}% <-this % stops a space
\thanks{The authors are with the Department of Electrical and Computer Engineering,
        University of California San Diego, 9500 Gilman Dr., La Jolla CA, 92093. Emails:
        {\tt\small lcoviell@ucsd.edu, massimo@ece.ucsd.edu}}%
}

\date{}

\begin{document}

\maketitle

\begin{abstract}
We consider a scenario in which leaders are required to recruit teams of followers. Each leader cannot recruit all  followers, but interaction is constrained according to a bipartite network.
The objective for each leader is to reach a state of \emph{local stability} in which it controls a team whose size is equal to a given constraint.
We focus on distributed strategies, in which agents have only local information of the network topology and propose a distributed algorithm in which leaders and followers act according to simple local rules. The performance of the algorithm is analyzed with respect to the convergence to a \emph{stable solution}.

Our results are as follows. For any network, the proposed algorithm is shown to converge to an \emph{approximate} stable solution in polynomial time, namely the leaders quickly form teams in which the total  number of additional followers required to satisfy all team size constraints is an arbitrarily small fraction of the entire population. 
In contrast, for general graphs there can be an exponential time gap between convergence to an approximate solution and to  a stable solution.
\end{abstract}

\section{Introduction}
A multi-agent system (MAS) is composed of many interacting intelligent agents. Agents can be software, robots, or humans, and the system is highly distributed, as agents  do not have a global view of the state and act autonomously of each other. These systems can be used to collectively solve problems  that are difficult to solve by a single entity. Their application ranges from robotics, to disaster response,  social structures, crowd-sourcing etc.
A main feature of MAS is that they can manifest self-organization as well as other complex control paradigms  even when the individual strategies of the agents are  very simple. In short, simple local interaction can conspire to determine complex global behaviors. Examples of such emerging  behaviors are in economics and game theory, where local preferences translate into global equilbria~\cite{roth}, in social sciences, where local exposure governs the spread of innovation~\cite{young2006diffusion}, 
 and in control, where local decision rules determine whether and how rapidly consensus is reached~\cite{blondel2005convergence,nedic2010convergence,nedic2010constrained,olshevsky2006convergence,tahbaz2008necessary,tahbaz2010consensus}.
 
From a practical perspective, the performance of a MAS often depends on how quickly convergence to a global, possibly approximate, solution is reached and it is in general influenced by the network structure. For example,  in the context of information diffusion in social networks, the rate of convergence of the system's dynamics is affected by the underlaying network and the local interaction rules~\cite{kleinberg2007cascading,montanari}.

One of the critical issues in multi-agent systems is coordination. Due to the autonomous behavior of the agents and to the absence of a central controller, coordination  must be distributed.
In the case of human agents, it is also important that the distributed control algorithm is simple enough to be suitable to model basic principles of human behavior~\cite{coviello}.
Two prominent  problems related to consensus and coordination in multi-agent systems are leader election and group formation. In the former case, multiple agents elect a leader that can then assign  tasks~\cite{nancy1996distributed}, while in the latter they divide themselves into teams in such a way that each agent knows to what team it belongs~\cite{franceschetti2001group}.
In both cases agents are all equal and coordination occurs among agents of a single class.
\begin{figure}
\centering
\psfrag{}[posn=c][psposn=c]{}
\includegraphics[scale=0.6]{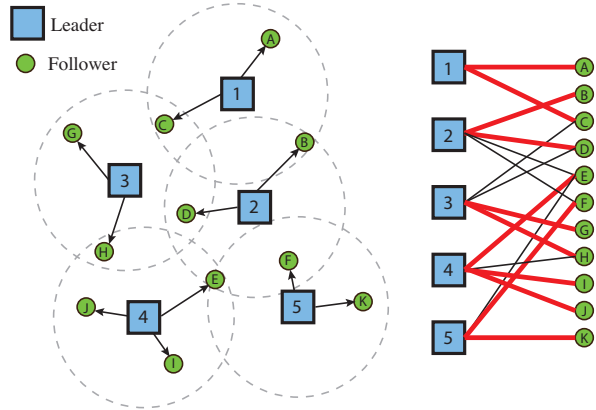}
\caption{
Example of a bipartite network between leaders and followers determined by physical constraints.
Left: each leader can only recruit the followers in its visibility range (dotted circle), arrows represent team membership, and the set of arrows defines a partition of the followers into teams.\newline
Right: the resulting bipartite network. An edge between leader $\ell$ and follower $f$ exists if and only if $f$ is in $\ell$'s visibility range. Matching edges define team membership and are highlighted.}
\label{fig:network}
\end{figure}

We consider a scenario in which there are agents of two classes, \emph{leaders} and \emph{followers}. Each leader must recruit a team of followers whose size is equal to a given constraint, by sending requests to the followers. Followers can only accept or reject incoming leaders' requests.
While multiple followers can be part of a leader's team, each follower can be part of a single team at any time, but is allowed to change team over time.
Moreover, a leader cannot recruit \emph{all} followers, but can only recruit the followers it is in direct communication with. The communication structure between leaders and followers is captured by an arbitrary bipartite network, and we assume that each agent has knowledge of and can interact with its neighbors only. That is, agents only have local knowledge of the underlying network. 
In general, the communication constraints of the population (and therefore the structure of the bipartite network) can be dictated by physical constraints (as for example antenna visibility range or signal to noise ratio threshold), social context, and so on.
A pictorial representation of a bipartite network arising from physical constraints is given in Figure~\ref{fig:network}.

We consider a notion  of stability in which each agent controls a team of adequate size. Each leader has an incentive to reach \emph{local} stability (that is, to build a team of followers of the right size) by dynamically interacting with its neighbors. The question we aim to answer is: can simple local rules lead to stable, or \emph{close} to stable, team formation  in reasonable time?
By ``close to stable'' we mean that the total number of additional followers required to satisfy all team size constraints is an arbitrary small fraction of the entire population.
We propose a simple, distributed, memoryless algorithm in which leaders do not communicate between each other, and we show that, in any network of size $n$, any constant approximation of a stable outcome (or of a suitably defined \emph{best} outcome if a stable one does not exist) is reached in time polynomial in $n$ with high probability. %(Theorem~\ref{thm:approx}). 
%This means that leaders quickly form teams in which the total number of additional followers required to satisfy all team size constraints is an arbitrary small fraction of the entire population.
In contrast, for general graphs we show through a counterexample that there can be an exponential gap between the time needed to reach stability and that needed to reach approximate stability, that is, to find the \emph{best} solution compared to a \emph{good} solution. 
We remark that, in its simplicity, the proposed algorithm is suitable to model human agents, it can be programmed on simple robots with limited computation abilities, and it is amenable to analysis.

The rest of the paper is organized as follows. After discussing how our work relates to the existing literature, in Section~\ref{sec:problem} we formally define the problem and the notions of stability and approximate stability, in Section~\ref{sec:algorithm} we present the distributed algorithm for leaders and followers, in Sections \ref{sec:approx} and \ref{sec:stable} we present our technical results on the algorithm's performance, and in Section~\ref{sec:simulations} we further discuss  the algorithm's performance by showing some simulations' results.
To prove our result on the convergence to approximate stability, we derive a technical lemma (Lemma~\ref{lem:karp}) that relates the quality of a matching to the existence of particular paths (that we call \emph{deficit-decreasing} paths) of given length. The lemma extends a known combinatorial result by Hopcroft and Karp~\cite{hopcroft} to the setup of many-to-one matching, and can be considered to be of independent interest.

\subsection{Related work}
\label{sec:related}
The problem of team formation that we consider is an example of distributed many-to-one matching in bipartite networks~\cite{ashlagi2011ec,hatfield2011,roth1984evolution}.
%In particular, given the set $L$ of leaders and the set $F$ of followers,  we are interested in finding connected subsets of $L\cup F$ of adequate size such that each element of $L\cup F$ is in exactly one subset and each subset contains a unique element of $L$.
%Since the seminal work on the Stable Marriage problem by Gale and Shapley~\cite{gale}, matching has drawn the attention of disciplines ranging from economics~\cite{hatfield2005matching,ostrovsky2008stability,roth1984evolution}, to engineering~\cite{cenedese2010decentralized}, health sciences~\cite{roth2003kidney} and computer science~\cite{hopcroft,israeli86,karp,lotker,mulmuley,pettie2004simpler}.
%The interested reader is referred to the book by Roth and Sotomayor~\cite{roth} for a general game-theoretic formulation of the problem.
The one-to-one case has been previously studied in the context of theoretical computer science~\cite{lotker}  \cite{pettie2004simpler}.
In the  control literature, our work is related to the distributed assignment problem and to group formation in MAS. In this framework, Moore and Passino~\cite{moore2007distributed} proposed a variant of the distributed auction algorithm for the assignment of mobile agents to tasks. Cenedese et al.~\cite{cenedese2010decentralized} proposed a variant of the Stable Marriage algorithm~\cite{gale} to solve the distributed task assignment problem. Abdallah and Lesser~\cite{abdallah2004organization} proposed an ``almost'' distributed algorithm for coalition formation, allowing for a special agent with the role of ``manager''. Gatson and den Jardins~\cite{gaston2005agent} studied a scenario of group formation where agents can adapt to the network structure. Tosic and Agha~\cite{tosic2004maximal} proposed an algorithm for group formation based on the distributed computation of maximal cliques in the underlying network. Further work studied team formation in multi-robot systems~\cite{vig2006multi}, in the case where communication between agents is not allowed~\cite{berman2009optimized}. Other authors considered MAS composed by leaders and followers. To cite a few, Tanner~\cite{tanner2004controllability} derived a necessary and sufficient condition for a group of interconnected agents to be controllable by one of them acting as a leader; Rahmani et al.~\cite{rahmani2009controllability} studied the controlled agreement problem in networks in which certain agents have leader roles, translating graph-theoretic properties into control-theoretic properties; Pasqualetti et al.~\cite{pasqualetti2008steering} analyzed the problem of driving a group of mobile agents, represented by a network of leaders and followers, in which follower act according to a simple consensus rule.

We distinguish ourselves from all mentioned papers, as we propose a fully distributed algorithm for group formation on arbitrary networks in which agents act according to simple local rules and perform very limited computation, and we derive performance guarantees in the form of theorems.
For an exhaustive overview on distributed algorithms in multi-agent systems, the interested reader is referred to the books by Lynch~\cite{nancy1996distributed} and by Bullo et al.~\cite{bullo2009distributed} and the references therein, while the survey by Horling and Lesser~\cite{horling} offers an overview on three decades of research on organizational paradigms as team and coalition formation.

A more recent line of research aims to study how humans connected over a network solve tasks in a distributed fashion~\cite{coviello,enemark,kearns2,kearns3,kearns,mccubbins}.
In the  work of Kearns et al.~\cite{kearns}, human subjects positioned at the vertices of a virtual network were shown to be able to collectively reach a coloring of the network, given only local information about their neighbors. Similar papers further investigated human coordination in the case of coloring~\cite{enemark,kearns2,mccubbins} and consensus~\cite{kearns2,kearns3}, with the main goal of characterizing how performance is affected by the network's structure.
Using experimental data of maximum matching games performed by human subjects in a laboratory setting, Coviello et al.~\cite{coviello} proposed a simple algorithmic model of human coordination that allows complexity analysis and prediction.
%In the same spirit, we are interested in designing simple local rules based on limited computation and information, with the goal of solving a task in a distributed fashion.

Finally, related to our work is also the research on social exchange networks~\cite{cook1983distribution,kleinberg2008balanced}, that considers a networked scenario in which each edge is associated to an economic value, nodes have to come to an agreement on how to share these values, and each agent can only finalize a single mutual exchange with a single neighbor. Recently, Kanoria et al.~\cite{kanoria2011fast} proposed a distributed algorithm that reaches approximate stability in linear time. However, we consider a different setup since we allow leaders to build teams of multiple followers.

%Carley~\cite{carley2002} observed that systems composed by computational agents are inherently computational, and therefore distributed algorithms and complexity analysis appear valuable tools to study the local and global dynamics of such systems.

\section{Problem formulation}
\label{sec:problem}
We consider a population composed of agents of two different classes: leaders and followers. Each leader is required to recruit a team of followers whose size is equal to a given constraint, by sending requests to the followers. Followers can only accept or reject leaders' requests. While multiple followers can be in a leader's team, each follower can be part of a single team at a time, but is allowed to change team over time.
A leader is not allowed to recruit \emph{all} followers, but can only recruit the followers it is in direct communication with.
The communication constraints of the population are captured by a bipartite network $G=(L\cup F,E)$ whose nodes' partition is given by the set $L$ of leaders and the set $F$ of followers, and where there exists an edge $(f,\ell)\in E$ between follower $f$ and leader $\ell$ if and only if $f$ and $\ell$ can communicate between each other (see Figure~\ref{fig:network}). Let $N_{\ell}=\{f\in F:(f,\ell)\in E\}$ be the neighborhood of $\ell\in L$.
For each $\ell\in L$, leader $\ell$ is required to recruit a team of  $c_{\ell}$ followers, where $c_{\ell}\ge 1$.
\begin{defn}[Matching]
A subset $M\subseteq E$ is a matching of $G$ if for each $f\in F$ there exists at most a single $\ell\in L$ such that $(\ell,f)\in M$.
\end{defn}

The definition of matching is consistent with the fact that multiple followers can be part of a leader's team.
There is a one-to-one correspondence between matchings $M$ of $G$ and tuples of teams $\{T_{\ell}(M):\ell\in L\}$, where $T_{\ell}(M)$ denotes the team of leader $\ell$ under the matching $M$. We have that $T_{\ell}(M)=\{f\in F: (\ell,f)\in M\}\subseteq N_{\ell}$ for every matching $M$.
We consider the following notion of stability.
\begin{defn}[Stable matching]
Given constraints $c_{\ell}$ for each $\ell\in L$, a matching $M$ of $G$ is stable if and only if $\vert T_{\ell}(M)\vert = c_{\ell}$ for all $\ell\in L$.
\end{defn}

Depending on the constraints $c_{\ell}$, a network $G$ might not admit a stable matching.
Nonetheless, given a matching of $G$, we are interested in assessing its \emph{quality}.
Our main result builds on the following definitions of \emph{deficit} of a leader and deficit of a matching.
\begin{defn}[Deficit of a leader]
Let $\ell$ be a leader with constraint $c_{\ell}\ge 1$, and $M$ be a matching of $G$. The deficit of $\ell$ under the matching $M$ is
$$d_{\ell}(M) = c_{\ell}-\vert T_{\ell}(M)\vert.$$
\end{defn} 
\begin{defn}[Deficit of a matching]
Given constraints $c_{\ell}\ge 1$ for each $\ell\in L$, the deficit of a matching $M$ of $G$ is
$$d(M) = \sum_{\ell\in L} d_{\ell}(M) = \sum_{\ell\in L}\left(c_{\ell}-\vert T_{\ell}(M)\vert\right).$$
\end{defn} 

In words, $d_{\ell}(M)$ is the number of additional followers leader $\ell$ needs to satisfy its size constraint.
Similarly, $d(M)$ sums the numbers of additional followers each leader needs to satisfy its size constraint.
Given a matching $M$, we say that a leader $\ell$ is \emph{poor} if $d_{\ell}(M)>0$ (that is, $\vert T_{\ell}(M)\vert < c_{\ell}$) and \emph{stable} if $\vert T_{\ell}(M)\vert = c_{\ell}$. In this work, we do not consider the case of $\vert T_{\ell}(M)\vert > c_{\ell}$ since we assume that each leader $\ell$ never recruits more than $c_{\ell}$ followers simultaneously. This can be justified by the fact that recruiting additional followers might be costly.

Observe that only poor leaders contribute to $d(M)$, and that $M$ is stable if and only if $d(M)=0$.
Given $G$, two matchings of $G$ can be compared with respect to their deficit, and the best matching of $G$ can be defined as one minimizing the deficit.
\begin{defn}[Best matching]
A matching $M$ of $G$ is a best matching of $G$ if $d(M)\le d(M')$ for every matching $M'$ of $G$.
\end{defn}

Observe that a stable matching is also a best matching.
Moreover, if $G$ admits a stable matching, $d(M)$ quantifies how much $M$ differs from a stable matching of $G$. In general, if $M^*$ is a best matching of $G$ with $d(M^*)=d^*$, then, $d(M)-d^*$ tells how much $M$ differs from a best matching of $G$.
Given a matching $M$ of $G$, the following definition provides a measure of how well $M$ approximates a best matching of $G$.
\begin{defn}[Approximate best matching]
Fix $\varepsilon\in[0,1]$, and let $m$ be the  number of followers in $G$.
Let $M^*$ be a best matching of $G$. Then, a matching $M$ is a $(1-\varepsilon)$-approximate best matching of $G$ if $d(M)-d(M^*)<\varepsilon m$.
\end{defn}

When $G$ admits a stable matching, we are interested in the notion of approximate stable matching.
\begin{defn}[Approximate stable matching]
Let $G$ admit a stable matching.
Fix $\varepsilon\in[0,1]$, and let $m$ be the  number of followers in $G$.
Then, a matching $M$ is a $(1-\varepsilon)$-approximate stable matching of $G$ if $d(M)<\varepsilon m$.
\end{defn}

\section{The algorithm}
\label{sec:algorithm}
We now present a distributed algorithm for team formation. Time is divided into rounds, and each round is composed by two stages. In the first stage, each leader acts according to the algorithm in Table~\ref{alg:leader}, and in the second stage each follower acts according to the algorithm in Table~\ref{alg:follower}.

First consider a leader $\ell$, and let $M$ be the matching at the beginning of a given round. If $\ell$ is poor (that is, $\vert T_{\ell}(M)\vert < c_{\ell}$) and $\vert T_{\ell}(M)\vert < \vert N_{\ell}\vert$  (that is, $\ell$ is not already matched with all followers in $N_{\ell}$) then, with probability $p$ (where $p\in(0,1]$ is a fixed constant), $\ell$ attempts to recruit an additional follower, chosen as explained below, by sending a \emph{matching request}. 
An unmatched follower in $N_{\ell}$, if any, is chosen uniformly at random; otherwise, a follower in $N_{\ell}\backslash T_{\ell}(M)$ is chosen uniformly at random. In other words, leaders always prefer to recruit followers that are currently unmatched over matched ones. Note that a leader tries to recruit an additional follower after checking if \emph{local stability} holds (that is, after checking if its team size is equal to $c_{\ell}$).

Consider now a follower $f$.
During each round, if $f$ has incoming  requests then each request is rejected independently of the others with probability $1-q$ (where $q\in(0,1]$ is a fixed constant). If all incoming requests are rejected, then $f$ does not change team (if currently matched) or it remains unmatched (if currently unmatched). Otherwise, one among the active requests is chosen uniformly at random, $f$ joins the corresponding leader, and all the other requests are discarded. For ease of presentation, we assume that a follower is equally likely to join a team when unmatched and to change team when currently matched, but all our results hold if we consider different values of $q$ for matched and unmatched followers (and even if we consider a different value of $q$ for each follower, as long as each value is a constant).

\floatname{algorithm}{Table}
\begin{algorithm}[h!]
\caption{Algorithm for leader $\ell\in L$}
\label{alg:leader}
\begin{algorithmic}
\IF{$\vert T_{\ell}(M)\vert < \min\{c_{\ell},\vert N_{\ell}\vert\}$}
\STATE with probability $p$ do the following
\IF{$\exists$ unmatched $f\in N_{\ell}$}
\STATE{choose an unmatched follower $f'\in N_{\ell}$ u.a.r.}
\ELSE
\STATE{choose a follower $f'\in N_{\ell}\backslash T_{\ell}(M)$ u.a.r.}
\ENDIF
\STATE{send a matching request to $f'$}
\ENDIF
\end{algorithmic}
\end{algorithm}

\floatname{algorithm}{Table}
\begin{algorithm}[h!]
\caption{Algorithm for follower $f\in F$}
\label{alg:follower}
\begin{algorithmic}
\IF{$f$ has incoming requests}
\FOR{each leader $\ell$ requesting $f$}
\STATE{with probability $1-q$ reject $\ell$'s request}
\ENDFOR
\IF{there are active requests}
\STATE{select one u.a.r. and join the corresponding team}
\STATE{reject all other requests}
\ENDIF
\ENDIF
\end{algorithmic}
\end{algorithm}

The proposed algorithm enjoys several properties. It is memoryless, the actions of each agent only depend on local information, and the leaders do not communicate between each other. 
Also, it is \emph{self-stabilizing}, that is, once a stable matching is reached, leaders stop recruiting followers.
Moreover, it is a single-stage algorithm, that is, agents never change their behavior until stability is reached.
Finally, observe that the exchanged messages can be represented by a single bit.

%The collaboration between leaders is minimal and consists in a leader attempting to recruit new followers only when its team size constraint is not met or there are unmatched followers in its neighborhood.

\section{Convergence to approximate stable matchings}
\label{sec:approx}
In this section, we only consider networks admitting stable matchings, and we show that, given any network and any constant $\varepsilon\in(0,1)$, a $(1-\varepsilon)$-approximate stable matching is reached in a number of rounds that is polynomial in the network size with high probability.
The assumption that a stable matching exists is for ease of presentation, and all our results also hold  for reaching approximate best matchings, by replacing $d(M)$ with $d(M)-d(M^*)$, where $M^*$ is a best matching of $G$.

Given a network $G$, for every $t\ge 0$, let $M(t)$ be the matching of $G$ at the beginning of round $t$, with deficit $d(M(t))$.
The next property follows from the fact that leaders do not voluntarily disengage from the followers in their teams (and therefore the deficit of a leader increases of a unit only if the deficit of another leader decreases by one unit).
\begin{cla}\label{cla:nonincr}
For $t\ge 0$, $d(M(t))$ is non-increasing in $t$.
\end{cla}

The next property follows from the assumption $c_{\ell}\ge 1,\forall\ell$.
\begin{cla}\label{cla:bound}
If $G$ admits a stable matching, then $d(M(t))\le m$ for every $t\ge 0$.
\end{cla}

We are now ready to state our main result.
\begin{thm}\label{thm:approx}
Let $G$ be a network with $m$ followers and which admits a stable matching. Let $\Delta=\max_{\ell\in L}\vert N_{\ell}\vert$ be the maximum degree of the leaders. Fix $0<\varepsilon<1$, and let 
$c\ge 1+\frac{1}{m(1-\varepsilon)}$.
Then, a $(1-\varepsilon)$-approximate stable matching of $G$ is reached within $c\lfloor 1/\varepsilon\rfloor(\Delta/pq)^{\lfloor 1/\varepsilon\rfloor}m$ rounds of the algorithm with probability at least $1-e^{-cm\varepsilon^2/2}$.
\end{thm}
%\begin{thm}\label{thm:approx}
%Let $G$ be a network with $m$ followers and which admits a stable matching. Let $\Delta=\max_{\ell\in L}\vert N_{\ell}\vert$ be the maximum degree of the leaders. Fix $0<\varepsilon<1$, and let $c(\varepsilon)=\lfloor 1/\varepsilon\rfloor (pq)^{-\lfloor 1/\varepsilon\rfloor}$ and $\alpha(\varepsilon)=\lfloor 1/\varepsilon\rfloor$.
%Then, a $(1-\varepsilon)$-approximate stable matching of $G$ is reached within $c(\varepsilon)\Delta^{\alpha(\varepsilon)}m^2$ rounds of the algorithm with probability at least $1-me^{-m/8}$.
%\end{thm}
\begin{example}
If $\Delta$ is constant in the network size, then one can choose $\varepsilon=1/\log m$, and Theorem~\ref{thm:approx} implies that a $(1-1/\log m)$-approximate stable matching is reached in at most $\mathcal{O}(m^2\log m)$ rounds with probability that goes to one as $m\to\infty$.
\end{example}

To prove Theorem~\ref{thm:approx}, we introduce the notion of \emph{deficit-decreasing} path, that in our setup plays the same role as the augmenting path in the context of one-to-one matching~\cite{diestel2005graph}.
Since we consider bipartite networks, a path alternates leaders and followers.
\begin{defn}[Deficit-decreasing path]
Given a matching $M$ of $G$, a cycle-free path $P=\ell_0,f_1,\ell_1,\ldots,f_k$ (of odd length 2k-1) is a deficit-decreasing path relative to $M$ if $(\ell_i,f_i)\in M$ for all $1\le i\le k-1$, $\ell_0$ is a poor leader, and $f_k$ is an unmatched follower.
\end{defn}

In words, a deficit-decreasing path starts at a poor leader with an edge not in $M$, ends at a follower that is not matched, and alternates edges in $M$ and edges not in $M$. To justify the nomenclature, observe that, if $d(M)>0$ and $P$ is a deficit-decreasing path relative to $M$, a new matching $M'$ such that $d(M')=d(M)-1$ can be obtained by flipping each unmatched edge of $P$ into a matched edge, and vice versa. This is depicted in Fig.~\ref{fig:ddpath}.
\begin{figure}
\centering
\psfrag{}[posn=c][psposn=c]{}
\includegraphics[scale=0.50]{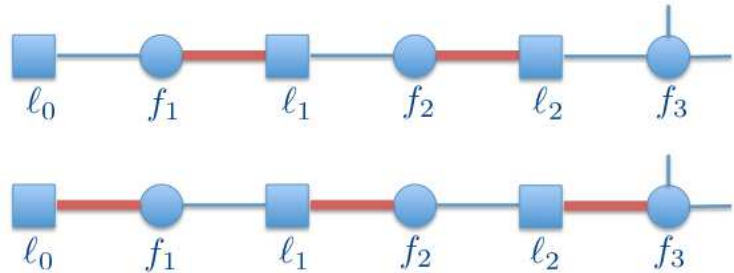}
\caption{A deficit-decreasing path of length $5$ is represented at the top of the figure: $\ell_0$ is a poor leader, $f_3$ is an unmatched follower, and matching edges are highlighted. The path is ``solved'' by turning each matched edge into an unmatched edge and vice versa, as show at the bottom of the figure: $\ell_0$ obtains an additional follower (and therefore its deficit decreases by a unit) and both $\ell_1$ and $\ell_2$ do not change their numbers of followers.}
\label{fig:ddpath}
\end{figure}

The proof of Theorem~\ref{thm:approx} builds on a technical lemma that, given a matching $M$ with $d(M)\ge\varepsilon m$, guarantees the existence of a deficit-decreasing path of length at most $2\lfloor 1/\varepsilon\rfloor$. The existence of such a path allows us to bound the number of rounds needed for a one-unit reduction of the deficit.
Our technical lemma extends a known result by Hopcroft and Karp~\cite[Theorem $1$]{hopcroft} given in the context of one-to-one matching, but our proof is more subtle because leaders can be matched to multiple followers and can have different size constraints $c_{\ell}$. The symmetric difference of two sets $A$ and $B$ is defined as $A\oplus B=(A\backslash B)\cup(B \backslash A)$. Two paths are \emph{follower-disjoint} if they do not share any follower (even though they might share some leader).
\begin{lem}\label{lem:karp}
Let $G$ admit a stable matching $N$. Let $M$ be a matching of $G$ with deficit $d(M)>0$. Then, in $M\oplus N$ there are at least $d(M)$ follower-disjoint deficit-decreasing paths relative to $M$.
\end{lem}
\begin{proof}
See Appendix~\ref{app:karp}.
\end{proof}

We make use of Lemma~\ref{lem:karp} through the following corollary.
\begin{cor}\label{cor:path}
Let $G$ be a network with $m$ followers, admitting a stable matching $N$. Let $M$ be a matching of $G$ with deficit $d(M)\ge\varepsilon m$, for some $\varepsilon>0$. Then, in $M\oplus N$ there exists a deficit-decreasing path relative to $M$ of length at most $2\lfloor 1/\varepsilon \rfloor-1$.
\end{cor}
\begin{proof}
By Lemma~\ref{lem:karp}, if $d(M)\ge\varepsilon m$ and $N$ is a stable matching of $G$, then in $M\oplus N$ we can choose $\varepsilon m$ follower-disjoint deficit-decreasing paths relative to $M$, whose cumulative length is at most $2m$ (since they do not share followers and $G$ is bipartite).
Necessarily, one of them has length at most $2\lfloor 1/\varepsilon \rfloor-1$ (note that a deficit-decreasing path has odd length).
\end{proof}

We are now ready to present the proof of Theorem~\ref{thm:approx}.

\subsection{Proof of Theorem~\ref{thm:approx}}
Let $G$ be a network with $m$ followers and which admits a stable matching. Fix $0<\varepsilon<1$.
For $t\ge 0$, $M(t)$ denotes the matching at the beginning of round $t$.
For every $0<x\le 1$, let
$$\tau(x)=\min\Big\{t\ge 0:d(M(t))< xm\Big\}$$
be the first round at whose beginning the deficit is strictly smaller than $xm$.
%By Property~\ref{cla:nonincr}, $\tau(x_2)\ge\tau(x_1)$ if $x_2<x_1$.
We are interested in bounding $\tau(\varepsilon)$.

Consider any round $t\ge 0$. By Property~\ref{cla:bound}, $d(M(t))\le m$, and therefore there exists
$0<\varepsilon'\le 1$ such that $d(M(t))=\varepsilon' m$ (we assume $\varepsilon'>0$, since the case of $\varepsilon'=0$ is trivial).
%Observe that $\tau(\varepsilon')$ can be equivalently defined as
%$$\tau(\varepsilon') = \min\Big\{t >t_1:d(M(t)) < d(M(t_1))\Big\},$$
%that is, the first round after $t_1$ in which the deficit decreases.
The following lemma bounds the number of rounds $\tau(\varepsilon')-t$ needed for a one-unit reduction of the deficit.
Let $\Delta=\max_{\ell\in L}\vert N_{\ell}\vert$ be the maximum degree of the leaders in $G$.
\begin{lem}
\label{lem:bound}
Let $d(M(t))=\varepsilon'm$ for some $0<\varepsilon'\le 1$. Then
\begin{equation*}
\Pr\Big( \tau(\varepsilon')-t \le \lfloor 1/\varepsilon'\rfloor \Big) \ge \left( \frac{pq}{\Delta}\right)^{\lfloor 1/\varepsilon'\rfloor}.
\end{equation*}
\end{lem}
%\begin{lem}
%Let $\Delta$ be the maximum degree of the leaders in $G$. Let $d(M(t_1))=\varepsilon'm$, for some $0<\varepsilon'\le 1$. Let $c(\varepsilon')=\lfloor 1/\varepsilon'\rfloor (pq)^{-\lfloor 1/\varepsilon'\rfloor}$ and $\alpha(\varepsilon')=\lfloor 1/\varepsilon'\rfloor$. Then
%\begin{equation*}
%\Pr\Big( \tau(\varepsilon')-t_1 \le c(\varepsilon')\Delta^{\alpha(\varepsilon')}m\Big) \ge 1-e^{-m/8}.
%\end{equation*}
%\end{lem}

\begin{proof}
Let $h(t)\ge 1$ be the odd length of the shortest deficit-decreasing path relative to $M(t)$.
By Corollary~\ref{cor:path}, $h(t)\le 2\lfloor 1/\varepsilon' \rfloor-1$.
 We distinguish the cases of $h(t)=1$ and $h(t)\ge 3$.

First consider $h(t)=1$. With probability at least $pq/\Delta$ the deficit decreases by at least one unit during the next round of the algorithm.
Too see this, consider a deficit-decreasing path $\ell,f$. With probability at least $p/\Delta$, $\ell$ attempts to recruit $f$ and, conditional on this event, $f$ considers $\ell$'s proposal with probability $q$, resulting in the lower bound $pq/\Delta$.
%Observe that the number of poor leaders $\ell'\neq\ell$ such that $\ell'f_1$ is a deficit-decreasing path of length $1$ does not influence the bound $pq/\Delta$. Indeed, if multiple poor leaders propose to $f_1$, then $f_1$ considers each proposal independently with probability $q$, and then chooses one uniformly at random.

Now consider $h(t)\ge 3$, and let $P$ be a shortest deficit-decreasing path of length $h(t)$ ending at an unmatched follower $f$.
By the same argument as above, the length of $P$ decreases by one during the next round with probability at least $pq/\Delta$
(observe that, as long as $h(t)>1$, $f$ remains unmatched during round $t$ since $P$ is a deficit decreasing path of shortest length).

By independence of successive rounds of the algorithm and the bound $h(t)\le 2\lfloor 1/\varepsilon' \rfloor -1$, with probability at least $(pq/\Delta)^{\lfloor 1/\varepsilon' \rfloor}$, a sequence of $\lfloor 1/\varepsilon' \rfloor-1$ rounds reduces the length of $P$ to $1$ and then in one additional round $P$ gets ``solved'' and the deficit decreases by one unit.
\end{proof}

Consider consecutive phases of $\lfloor 1/\varepsilon \rfloor$ rounds each. For phases $i=0,1,2,\ldots$, let $X_i$ be $iid$ Bernoulli random variables with $\Pr(X_i=1)=( pq/\Delta)^{\lfloor 1/\varepsilon\rfloor}$.
By Lemma~\ref{lem:bound}, after $T$ phases (i.e., at the beginning of round $t^*=T\lfloor 1/\varepsilon \rfloor$), the deficit of the matching is upper bounded by
$$
d(M(t^*))<\max\left\{ \varepsilon m, m+1-\sum_{i=1}^{T}X_i \right\},
$$
since by Property~\ref{lem:bound} the matching at the beginning of round $0$ has deficit $d(M(0))\le m$.
By independence of the phases, a Chernoff bound implies that for any $0<\delta\le 1$
$$
\Pr\Big(\sum_{i=1}^{T}X_i<(1-\delta)T( pq/\Delta)^{\lfloor 1/\varepsilon\rfloor}\Big)<e^{-T( pq/\Delta)^{\lfloor 1/\varepsilon\rfloor}\delta^2/2}.
$$

Setting $\delta=\varepsilon$ and $T=cm( \Delta/pq)^{\lfloor 1/\varepsilon\rfloor}$ (where $c$ is a constant to be specified later), the deficit of the matching at the beginning of round $t^*=\lfloor 1/\varepsilon\rfloor cm( \Delta/pq)^{\lfloor 1/\varepsilon\rfloor}$ is upper bounded by
$$
d(M(t^*))<\max\left\{ \varepsilon m, m+1-(1-\varepsilon)cm \right\}
$$
with probability at least $1-e^{-cm\varepsilon^2/2}$. To conclude the proof of the theorem we need that $\varepsilon m\ge m+1-(1-\varepsilon)cm$, which is true for any $c\ge 1+ \frac{1}{m(1-\varepsilon)}$.

\section{Exponential convergence}
\label{sec:stable}
%Theorem~\ref{thm:approx} shows that, for any network $G$ admitting a stable matching, and for any $\varepsilon\in (0,1)$, our algorithm reaches a $(1-\varepsilon)$-approximate stable matching of $G$ in a number of rounds that is polynomial in the network size with high probabiltity. Moreover, if the maximum degree of the leaders is bounded, a polynomial bound holds for reaching a stable matching.
Theorem~\ref{thm:approx} gives a polynomial bound for reaching a $(1-\varepsilon)$-approximate stable matching for any constant $0<\epsilon<1$ and any network.
However, a similar guarantee cannot be derived for the case of a stable matching, as shown in this section through a counterexample.
In particular, we define a sequence of networks of increasing size and maximum degree that diverges with the network size, and show that the number of rounds required to converge from an approximate matching $M$ with $d(M)=1$ to the stable matching  (that is, to reduce the deficit of a single unit) is exponentially large in the network's size with high probability from an overwhelming fraction of the approximate matchings $M$ such that $d(M)=1$.

For $n\ge 1$, let $G_n=(L_n\cup F_n,E_n)$ be the network with $n$ leaders and $n$ followers (i.e., $L_n=\{\ell_1,\ldots,\ell_n\}$ and $F_n=\{f_1,\ldots,f_n\}$), with edges $E_n=\{(\ell_i,f_j):1\le i\le n,j\le i\}$, and team size constraints $c_{\ell}=1$ for all $\ell\in L_n$, see Figure~\ref{fig:bad}. $G_n$ has maximum degree $n$ and a unique stable matching given by $M_n^*=\{(\ell_i,f_i):1\le i\le n\}$.
\begin{figure}[h!]
\centering
\includegraphics[scale=0.4]{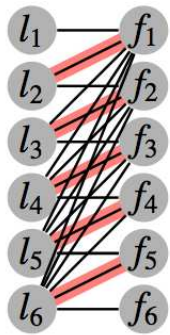}
 %\usebox{\bipartite}
\caption{The network $G_n$ for $n=6$. The matching $M_n'$ is highlighted.}
\label{fig:bad}
\end{figure}
\begin{thm}\label{thm:bad}
For any matching $M$ of $G_n$, let $\tau^*(M)$ denote the number of rounds to converge  to the perfect matching when starting from $M$.
Then, for any fixed constant $0<\gamma<1$, $\tau^*(M)$ is exponentially large in $\gamma n$ with high probability for a $1-O(n2^{-(1-\gamma)n})$ fraction of all the matchings $M$ such that $d(M)=1$. 
\end{thm}

%In words, Theorem~\ref{thm:bad} states that convergence to stability requires exponentially many rounds of the algorithm from an overwhelming fraction of all the matching in which a single leader is not satisfied of its team size.
Here we only provide a sketch of the proof, whose details are presented in Appendix~\ref{app:exp}.
%The key of the proof is to study the algorithm's dynamics starting from matchings $M$ with deficit $d(M)=1$, and to show that the time needed to converge to stability grows exponentially with the network's size.
To get an understanding of the algorithm's dynamics, consider the matching
$$ M_n'=\{(\ell_i,f_{i-1}): 2\le i\le n\},$$
highlighted in Figure~\ref{fig:bad} for the case of $n=6$.
Observe that $d(M_n')=1$ and, under $M_n'$, $\ell_1$ is poor, and the remaining leaders are stable. According to the algorithm, $\ell_1$ attempts to recruit  $f_1$ (currently in $\ell_2$'s team). If $f_1$ accepts, then $\ell_1$ becomes stable and $\ell_2$ becomes poor (and can in turn attempt to recruit either $f_1$ or $f_2$).
After each round, there exists a unique poor leader until the stable matching is reached. The stable matching is reached when $\ell_{n-1}$ ($\ell_5$ in Figure~\ref{fig:bad}) becomes poor and then successfully recruits $f_{n-1}$ ($f_5$ in Figure~\ref{fig:bad}), and finally $\ell_n$ successfully recruits $f_n$ (recall that leaders prefer unmatched followers).

In general, fix any matching $M$ of $G_n$ such that $d(M)~=~1$. In $M$, there is a single poor leader $\ell_{i_0}$ and a single unmatched follower $f_{i_K}$.
$M$ is associated to a unique deficit-decreasing path $\ell_{i_0},f_{i_0},\ldots,\ell_{i_{K-1}},f_{i_{K-1}},\ell_{i_K},f_{i_K}$.  We define the \emph{height} $h(M)$ of $M$ as follows. If $K\ge 1$ then $h(M)=i_{K-1}$, if $K=0$ then $h(M)=0$.

Starting from $M$, for every $t<\tau(M)$, the matching $M(t)$ at the beginning of round $t$ has deficit $d(M(t))=1$ (by Property~\ref{cla:nonincr}), a single poor leader denoted by $\ell_{i(t)}$, the single unmatched follower $f_{i_K}$ and height $h(M(t))=h(M)=i_{K-1}$.
The stochastic process $\{i(t)\}$ tracking the position of the poor leader $\ell_{i(t)}$ is not a classical random walk on $\{\ell_1,\ldots,\ell_{i_K}\}$ and its transition probabilities at each round depend on the current matching. The time to reach stability is upper bounded by $\min\{t:i(t)=h(M)\}$, that is, the first round in which $\ell_{h(M)}$ becomes poor (since $\ell_{h(M)}$ can then match with $f_{h(M)}$ leaving $\ell_{i_K}$ poor, who would in turn match with the unmatched follower $f_{i_K}$, thus reaching the stable matching).

We prove a one-to-one correspondence between the matchings $M(t)$ reachable from $M$ in which $i(t)\le h(M)$ (note that $d(M(t))=1$ for each of them) and the nodes of a tree whose size is exponentially large in the height $h(M)$. In particular, we can show that the process $\{M(t):t\ge 0,M(0)=M\}$ is equivalent to a classical random walk on the nodes of the tree, and that reaching the matching with $i(t)=h(M)$ corresponds to reaching the root of the tree. A random walk starting at any node of the tree visits the root after a number of steps that is exponentially large in the height $h(M)$ with high probability. Finally, the proof of Theorem~\ref{thm:bad} is completed by arguing that, for any constant $0<\gamma<1$, a $1-O(n2^{-(1-\gamma)n})$ fraction of all matchings $M$ of $G_n$ such that $d(M)=1$ have height $h(M)\ge\gamma n$.

\section{Simulations}
\label{sec:simulations}
In this section, the performance of our algorithm is further evaluated through simulation.
In Figure~\ref{fig:plotlog}, the algorithm's average convergence time on the sequence of networks $G_n$ defined in Section~\ref{sec:stable} is shown (in logarithmic scale).
On the one hand, the thick solid line suggests that the average number of rounds to reach a $0.9$-approximate stable matching is upper bounded by a polynomial of small degree, consistently with Theorem~\ref{thm:approx}. On the other hand, convergence to the stable matching requires an average number of rounds that grows exponentially in $n$ (thin solid line), as predicted by Theorem~\ref{thm:bad}.
Moreover, the dotted line represents the average time after which all followers become matched, that grows slowly with $n$.

Figure~\ref{fig:plotAppr} shows the algorithm's performance in reaching successively finer approximations of the best matching on random networks $G(n,m,\rho)$. Here, $G(n,m,\rho)$ refers to a random bipartite network with $n$ leaders and $m$ followers, in which each edge exists independently of the others with probability $\rho$ (we fixed $\rho=0.04$), and with constraint $c_{\ell}=\min\{m/n,\vert N_{\ell}\vert\}$ for each leader $\ell$ . For each of the $(n,m)$ pairs that we considered, $20$ random $G(n,m,\rho)$ were generated, and the algorithm was run $20$ times on each.
%The $x$-axis represents the approximation factor $1-\varepsilon$, while the $y$-axis represents the average of the number of rounds $\tau(\varepsilon)$ needed to reach a $(1-\varepsilon)$-approximate best matching of $G(n,m,p)$.
We observe that, consistently with Theorem~\ref{thm:approx}, $\tau(\varepsilon)$ increases both when $\varepsilon$ decreases (i.e., when a finer approximation is desired) and when the number $m$ of followers increases. The plot visually suggests that a good solution is reached quickly, while most of the time is spent in the attempt of improving it to the best solution.
\begin{figure}
\centering
\psfrag{A}[posn=c][psposn=c]{\small{$n$ }}
\psfrag{B}[posn=c][psposn=c]{\small{Rounds (log-scale)}}
\psfrag{approxmatchingxxxxxxxxxxxxxx}[posn=c][psposn=c]{\small{$0.9$-approximation}}
\psfrag{stablematchingxxxxxxxxxxxxxxxxxx}[posn=c][psposn=c]{\small{Stable matching}}
\psfrag{followersmatchedxxxxxxxxxxxx}[posn=c][psposn=c]{\small{Followers matched}}
\includegraphics[scale=0.5]{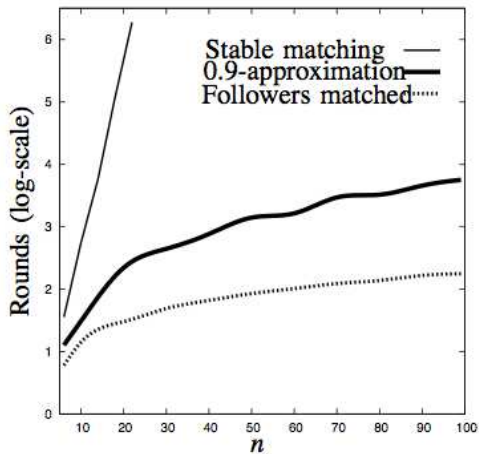}
\caption{Algorithm's convergence time on the sequence of networks $G_n$.}
\label{fig:plotlog}
\end{figure}

\begin{figure}
\centering
\psfrag{A}[posn=c][psposn=c]{\small{$1-\varepsilon$ }}
\psfrag{B}[posn=c][psposn=c]{\small{Rounds}}
\psfrag{n100m200xxxxxxxxxxxxxxxxxx}[posn=c][psposn=c]{\small{$n=100,m=200$}}
\psfrag{n100m300xxxxxxxxxxxxxxxxxx}[posn=c][psposn=c]{\small{$n=100,m=300$}}
\psfrag{n150m450xxxxxxxxxxxxxxxxxx}[posn=c][psposn=c]{\small{$n=150,m=450$}}
\psfrag{n200m600xxxxxxxxxxxxxxxxxx}[posn=c][psposn=c]{\small{$n=200,m=600$}}
\includegraphics[scale=0.5]{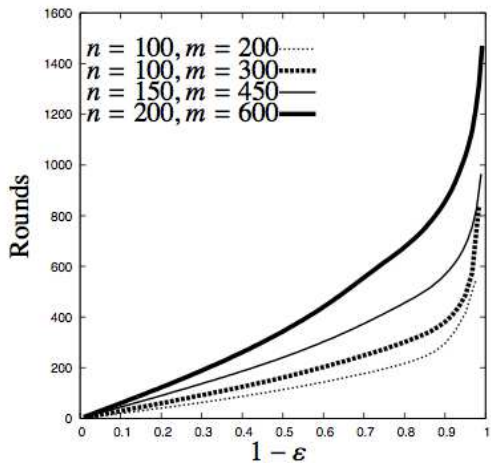}
\caption{Algorithm's average time to reach a $(1-\varepsilon)$-approximate best matching on random bipartite networks $G(n,m,\rho)$, for $\rho=0.04$.}
\label{fig:plotAppr}
\end{figure}

\section{Discussion}
The distributed algorithm we proposed, in which leaders and followers act according to simple local rules, is computationally tractable and allows us to derive performance guarantees in the form of theorems.
Despite its simplicity, the algorithm is shown to reach an arbitrarily close approximation of a stable matching (or of a best matching) in polynomial time in any network. However, in general there can be an exponential gap between reaching an approximate solution and a stable solution.

In the proposed algorithm, leaders do not communicate between each other, and only act in response to their own status and the status of their neighborhoods.
The only collaboration between them consists in the fact that the leaders whose size constraints are satisfied do not attempt to recruit additional matched followers, and this is justified since recruiting more followers might be costly.
How communication between leaders affects performance is an open question, as well as determining what amounts of communication and complexity are necessary to remove the exponential gap in the case of unbounded degree networks.

Finally, in Section~\ref{sec:stable}, we defined a sequence of networks in which the maximum degree of the leaders scales linearly with the network size.
It would be interesting to understand whether a counterexample in which the maximum degree scales more slowly (e.g., logarithmically in the network size) could be derived.

%\section*{Aknowledgement}
%This work is supported by 
%\appendix
\section*{Appendix}
\setcounter{section}{0}
\renewcommand{\thesection}{\Alph{section}}

\section{Proof of Lemma~\ref{lem:karp}}
\label{app:karp}
Given the matching $M$ and the stable matching $N$, for brevity we write deficit-decreasing path instead of deficit-decreasing path in $M\oplus N$ relative to $M$. Similarly, by telling that leader $\ell$ and follower $f$ are matched we mean that $(\ell,f)\in M$, unless otherwise specified.

We prove a stronger claim than the one stated in the lemma, proceeding as follows. First, we show that for each leader $\ell$ with deficit $d_{\ell}(M)>0$ there are at least $d_{\ell}(M)$ follower-disjoint deficit-decreasing paths starting at $\ell$. Then, we argue that $d(M)$ follower-disjoint  deficit-decreasing paths can be chosen, $d_{\ell}(M)$ of which start at each leader $\ell$ with deficit $d_{\ell}(M)>0$.

Consider a leader $\ell$ with $d_{\ell}(M)>0$. Assume by contradiction that there are strictly less then $d_{\ell}(M)$ follower-disjoint deficit-decreasing paths starting at $\ell$, and refer to Fig.~\ref{fig:poorleader} for a schematic representation.

Since $\ell$ has a team size constraint $c_{\ell}>0$, there are exactly $c_{\ell}-d_{\ell}(M)$ followers that are matched to $\ell$. Observe that no follower matched to $\ell$ can be the first follower of a deficit-decreasing path starting at $\ell$, since a deficit-decreasing path starts with an edge in $N\backslash M$.

Since $G$ admits a stable matching, the neighborhood $N_{\ell}$ of $\ell$ has size $\vert N_{\ell}\vert\ge c_{\ell}$. Therefore, there are are $k\ge d_{\ell}(M)$ followers in $N_{\ell}$ that are not matched to $\ell$.
Assume that $h<d_{\ell}(M)$ of the followers in $N_{\ell}$ are the first followers of $h$ follower-disjoint deficit-decreasing paths starting at $\ell$ (these paths are denoted by $P_1,\ldots,P_h$ in Figure \ref{fig:poorleader}).
Denote the remaining $k-h>0$ followers by $f_1,\ldots,f_{k-h}$, and assume by contradiction that none among  them is the first follower of a deficit-decreasing path starting at $\ell$ (this is equivalent to assuming that there are strictly less than $d_{\ell}(M)$ follower-disjoint deficit-decreasing paths starting at $\ell$).
\begin{figure}
\centering
\psfrag{}[posn=c][psposn=c]{}
\includegraphics[scale=0.50]{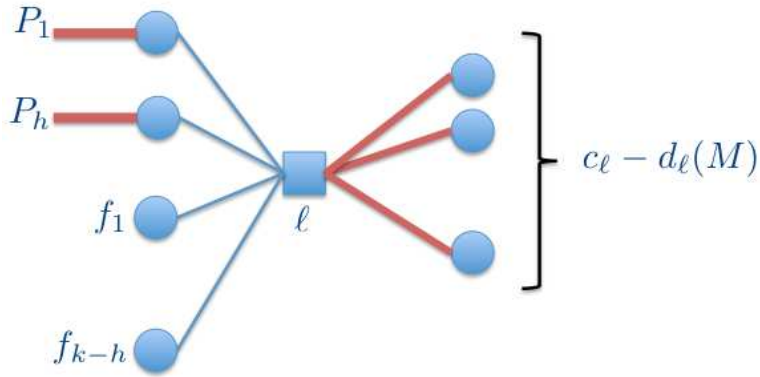}
\caption{A leader $\ell$ with constraint $c_{\ell}$, degree $\vert N_{\ell}\vert\ge c_{\ell}$ and deficit $d_{\ell}(M)$. Matched edges are highlighted. $\ell$ is matched to exactly $c_{\ell}-d_{\ell}(M)$ followers (depicted on the right). Among the other $k\ge d_{\ell}(M)$ followers in $N_{\ell}$, $h<k$ of them are the first followers on $h$ follower-disjoint deficit-decreasing paths starting at $\ell$ (these paths are denoted by $P_1,\ldots,P_h$), and none of the remaining $k-h$ (denoted by $f_1,\ldots,f_{k-h}$ ) is the first follower of a deficit-decreasing paths starting at $\ell$.}
\label{fig:poorleader}
\end{figure}

Observe that, in order to become stable, $\ell$ needs to match with at least one additional follower among $\{f_1,\ldots,f_{k-h}\}$.
We show that, under the assumption above, a one-unit reduction in the deficit of $\ell$ would eventually result in a one-unit increase of the deficit of another leader, implying that $G$ does not admit a stable matching, generating a contradiction.

Consider any follower $f'\in\{f_1,\ldots,f_{k-h}\}$, and observe that $f'$ is matched in $M$ since otherwise $\ell f'$ would be a deficit-decreasing path starting at $\ell$.
Let $\ell'$ be the leader such that $(\ell',f')\in M$, and observe that if $\ell'$ is matched to all followers in $N_{\ell'}$ then $\ell$ cannot match to $f'$ without causing a one-unit increase of the deficit of $\ell'$.
Therefore assume that in $N_{\ell'}$ there is a follower $f''$ such that $(\ell'',f'')\in M$ for some leader $\ell''\neq\ell'$ ($f''$ is matched in $M$ since otherwise $\ell,f',\ell',f''$ would be a deficit-decreasing path starting at $\ell$, see Fig.~\ref{fig:nopath}).  In the following two cases $\ell$ cannot match to $f'$ without eventually increasing the deficit of another leader.
\begin{itemize}
\item[(i)] $\ell''=\ell$. In this case $\ell,f',\ell',f'',\ell$ is a cycle, and if $\ell$ matches to $f'$ then the deficit of a leader in the cycle must increase of one unit.
\item[(ii)] $\ell''\neq\ell$ and $\ell''$ is matched to all followers in $N_{\ell''}$ other than $f'$. In this case if $\ell$ matches to $f'$ then the deficit of a leader on the path $\ell,f',\ell',f'',\ell''$ must eventually increase by a unit.
\end{itemize}
Therefore assume that in $N_{\ell''}$ there is a follower $f'''$ such that $(\ell''',f''')\in M$ for some leader $\ell'''\neq\ell''$ (again, $f'''$ is matched in $M$ since otherwise $\ell,f',\ell',f'',\ell'',f'''$ would be a deficit-decreasing path).
Again, $\ell$ cannot match to $f'$ without eventually increasing the deficit of another leader if either $\ell'''=\ell$ or $\ell'''=\ell'$ (each similar to the case (i) above), or if $\ell'''$ is matched to all followers in $N_{\ell''}$ other than $f',f''$ (similar to the case (ii) above).

By iteration, it follows that $\ell$ cannot match to any follower $f'\in\{f_1,\ldots,f_{k-h}\}$ without eventually increasing the deficit of another leader, in contradiction with the existence of the stable matching $N$. Hence, there are at least $d_{\ell}(M)$ follower-disjoint deficit-decreasing paths starting at $\ell$.
\begin{figure}
\centering
\psfrag{}[posn=c][psposn=c]{}
\includegraphics[scale=0.50]{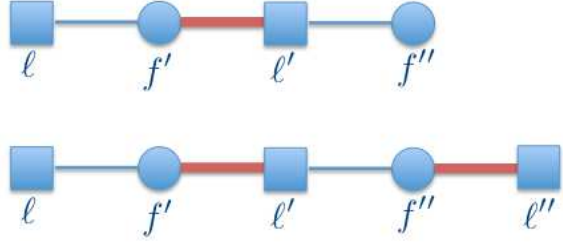}
\caption{If  $f''$ is not matched then $\ell,f',\ell',f''$ would be a deficit-decreasing path (shown at the top of the figure, in which matching edges are highlighted), contradicting the assumption that no follower in $\{f_1,\ldots,f_{k-h}\}$ can be the first follower of a deficit-decreasing path starting at $\ell$. Therefore, $f''$ is matched to a leader $\ell''$ (the bottom of the figure represents the case of $\ell\neq\ell''$).}
\label{fig:nopath}
\end{figure}

To complete the proof of the lemma, we show that we can choose $d(M)$ follower-disjoint deficit-decreasing paths, $d_{\ell}(M)$ of which start at each leader $\ell$ with $d_{\ell}(M)>0$.

We proceed by contradiction, and make the following assumption. Let $\mathcal{P}$ be any set of $d(M)$ deficit-decreasing paths, $d_{\ell}(M)$ of which start at each leader $\ell$ with $d_{\ell}(M)>0$ (denote by $\mathcal{P}_{\ell}$ the elements of $\mathcal{P}$ starting at $\ell$); then, there are two leaders $\ell$, $\ell'$ such that two paths $P\in\mathcal{P}_{\ell}$, $P'\in\mathcal{P}_{\ell'}$ are not follower-disjoint.
In order to reach the stable matching $N$ starting from $M$, a set of $d(M)$ deficit-decreasing paths must be solved. However, if $P$ is solved (by ``flipping'' matched edges into unmatched edges, and vice versa) then $P'$ is not solved, and if $P'$ is solved then $P$ is not solved (see Figures \ref{fig:intersectingdd1} and \ref{fig:intersectingdd2} for a schematic representation). If follows that $N$ cannot be reached from $M$ by solving the $d(M)$ deficit-decreasing paths in $\mathcal{P}$.

The last argument holds  for any choice of $\mathcal{P}$, and this generates a contradiction on the reachability of $N$ starting from $M$ (observe that $N$ can be reached from $M$ in finite time, e.g. by a cat-and-mouse argument on the space of all the matchings of $G$). 
Hence, we can choose $d(M)$ follower-disjoint deficit-decreasing paths, $d_{\ell}(M)$ of which start at each leader $\ell$ with $d_{\ell}(M)>0$, and the lemma is proven.
\begin{figure}
\centering
\psfrag{}[posn=c][psposn=c]{}
\includegraphics[scale=0.50]{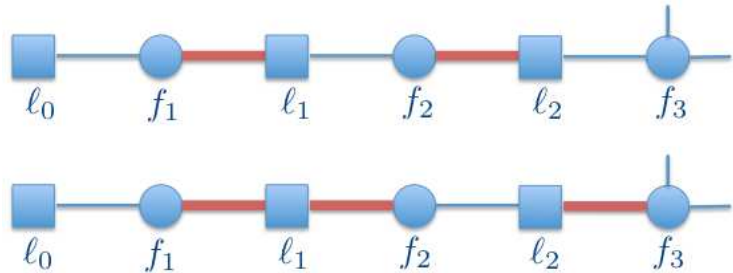}
\caption{Given the matching at the top of the figure (matching edges are highlighted), assume that both $\ell_0$ and $\ell_1$ are poor, and that $f_3$ is unmatched. The deficit-decreasing paths $P=\ell_0,f_1,\ell_1,f_2,\ell_2,f_3$ and $P'=\ell_0,f_1,\ell_1,f_2$ are not follower-disjoint. If $P'$ is solved (shown at the bottom of the figure), then $P$ is not  solved, and vice versa.}
\label{fig:intersectingdd1}
\end{figure}

\begin{figure}
\centering
\psfrag{}[posn=c][psposn=c]{}
\includegraphics[scale=0.50]{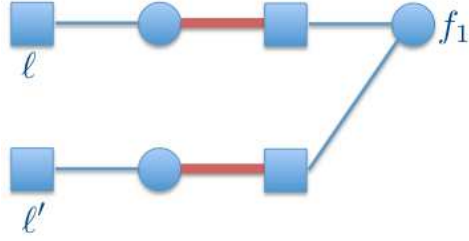}
\caption{If, under the matching highlighted in the figure, both $\ell$ and $\ell'$ are poor and $f_1$ is unmatched then there are two deficit-decreasing paths that are not follower-disjoint (one starting at $\ell$ ad ending at $f_1$, the other starting at $\ell'$ ad ending at $f_1$). If one of them is solved then the other is not solved, and vice versa.}
\label{fig:intersectingdd2}
\end{figure}

\section{Proof of Theorem~\ref{thm:bad}}
\label{app:exp}
Let $\M_n$ be the set of all the matchings of $G_n$ such that $d(M)=1$. We proceed as follows.
First, we show that each $M\in\M_n$ is uniquely identified by the set of the leaders that are not matched with ``horizontal'' edges (that is, leaders $\ell_i$ such that $(\ell_i,f_i)\notin M$).
Second, we define trees $T^*_m$, $m\ge 1$ such that a random walk on $T^*_m$ starting at any node different than the root hits the root after a number of steps that is exponentially large in $m$ with high probability.
Third, for each matching $M\in\M_n$ we define a quantity $h(M)$ that we call the \emph{height} of $M$ and we argue that, when initialized at $M$, the algorithm's dynamics is equivalent to a random walk on the tree $T^*_{h(M)}$ and reaching the stable matching of $G_n$ corresponds to reaching the root of $T^*_{h(M)}$ (and therefore it requires a number of rounds that is exponentially large in $h(M)$ with high probability).
Finally, by a counting argument, we show that for any constant $0<\gamma<1$ a $1-O(n2^{-(1-\gamma)n})$ fraction of all the matchings in $\M_n$ have height at least $\gamma n$, completing the proof of the theorem.

\subsection{Properties of the matchings in $\M_n$.}
Matchings in $\M_n$ enjoy the following structural properties.
\begin{lem}
\label{lem:Mprop}
Let $M\in\M_n$. The following properties hold.
\begin{itemize}
\item[(1)] There are a single poor leader $\ell_{i^*(M)}$ and a single unmatched follower $\ell_{j^*(M)}$ in $M$.
\item[(2)] $1\le i^*(M)\le j^*(M)\le n$.
\item[(3)] $(\ell_k,f_k)\in M$ for all $k<i^*(M)$ and all $k> j^*(M)$.
\item[(4)] Let $\I(M)=\{j_0,j_1,\ldots,j_K\}$ be the sorted set of indexes $j$ such that $(\ell_j,f_j)\notin M$. Then
  \begin{itemize}
  \item[(a)] $j_1=i^*(M)$ and $j_K=j^*(M)$.
  \item[(b)] $(\ell_{j_{k+1}},f_{j_k})\in M$ for all $k\in\{0,\ldots,K-1\}$.
  \end{itemize}
\end{itemize}
\end{lem}

\begin{proof}
Property (1). Since $d(M)=\sum_{\ell\in L}d_{\ell}(M)=1$, there is a single poor leader $\ell_{i^*(M)}$ in $M$. Since $c_{\ell}=1$ for all $\ell\in L$, each leader $\ell\neq \ell_{i^*(M)}$ is matched to a single follower. It follows that there is a unique unmatched follower $f_{j^*(M)}$.

Property (2). Suppose by contradiction that $i^*(M)>j^*(M)$. Since $N_{\ell_{j^*(M)}}=\{f_1,\ldots,f_{j^*(M)}\}$ and $f_{j^*(M)}$ is unmatched, leader $\ell_{j^*(M)}$ is matched to one of the followers in $\{f_1,\ldots,f_{j^*(M)-1}\}$.
Hence, the $j^*(M)-1$ leaders $\ell_1,\ldots,\ell_{j^*(M)-1}$ are matched to at most $j^*(M)-2$ out of the $j^*(M)-1$ followers  $f_1,\ldots,f_{j^*(M)-1}$, and one of them is necessarily poor, contradicting Property (1). Therefore, $i^*(M)\le j^*(M)$.

Property (3). We proceed by induction. If $i^*(M)>1$, then $(\ell_1,f_1)\in M$ since $N_{\ell_1}=\{f_1\}$ and $\ell_1$ is matched with a follower. Assume that if $i^*(M)>j$ then $(\ell_k,f_k)\in M$ for all $k\le j$. If $i^*(M)>j+1$, then, by the inductive assumption, $\ell_{j+1}$ can only be matched to $f_{j+1}$ since $N_{\ell_{j+1}}=\{f_1,\ldots,f_{j+1}\}$. This shows that $(\ell_k,f_k)\in M$ for all $k<i^*(M)$.
If $j^*(M)< n$ then $(\ell_n,f_n)\in M$ since $f_n$ is matched and $\ell_n$ is the only leader connected to $f_n$. Assume by induction that if $j^*(M)<j$ then $(\ell_k,f_k)\in M$ for all $k\ge j$. If $j^*(M)<j-1$, then, by the inductive assumption, $f_{j-1}$ can only be matched to $\ell_{j-1}$ since $f_{j-1}$ is adjacent to $\ell_{j-1},\ldots,\ell_n$. This shows that $(\ell_k,f_k)\in M$ for all $k>j^*(M)$.

Property (4). If $K=0$ then $M=\{(\ell_i,f_i):i\neq i^*(M)\}$, $j^*(M)=i^*(M)$, and properties (4a) and (4b) trivially hold.
Now consider $K\ge 1$. Let $\I(M)=\{j_0,j_1,\ldots,j_K\}$ be the sorted set of indexes $j$ such that $(\ell_j,f_j)\notin M$.
By property (3), we have that $j_0=i^*(M)$ and $j_K=j^*(M)$, therefore property (4a) follows. Hence, $(\ell_{j_2},f_{j_1})\in M$ 
since $(\ell_k,f_k)\in M$ for all $k\in\{j_1+1,\ldots,j_2-1\}$ by definition of $\I(M)$, and $N_{\ell_{j_2}}=\{f_1,\ldots,f_{j_2}\}$.
Property (4b) follows by induction.
\end{proof}

Lemma~\ref{lem:Mprop} states that non-horizontal matching edges do not intersect.
In particular, given a matching $M\in\M_n$, the set $\I(M)$ represents the set of (the sorted indexes of) the leaders that are not matched with horizontal edges (see Figure~\ref{fig:set_i} for an example), $\ell_{i^*(M)}$ for $i^*(M)=\min\I(M)$ is the unique unmatched leader, and $\ell_{j^*(M)}$ for $j^*(M)=\max\I(M)$ is the unique unmatched follower. Recall that $M_n^*=\{(\ell_k,f_k):1\le k\le n\}$ is the unique stable matching of $G_n$, and let $\I(M_n^*)=\emptyset$.
Lemma~\ref{lem:Mprop} implies that every matching $M\in\M_n\cup\{M^*n\}$ is uniquely identified by the set $\I(M)$. In particular, the following result holds.
\begin{figure}[h!]
\centering
% \usebox{\nonhorizontal}
\includegraphics[scale=0.4]{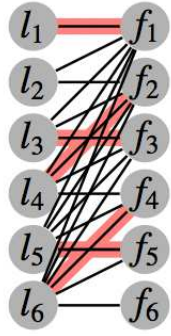}
\caption{An example of a matching $M$ of $G_6$ with $d(M)=1$. $M$ is uniquely determined by the set $\I(M)=\{2,4,6\}$, that encodes the following: $\ell_2$ is not matched, $\ell_4$ is matched with $f_2$, $\ell_6$ is matched with $f_4$, $f_6$ is not matched. Also note that $P(M)=\ell_2,f_2,\ell_4,f_4,\ell_6,f_6$ is the unique deficit-decreasing path relative to $M$.}
\label{fig:set_i}
\end{figure}
\begin{lem}
\label{lem:M121}
Consider the mapping $\I(\cdot)$ from $\M_n\cup\{M^*_n\}$ to $\s = \big\{A:A\subseteq\{1,\ldots,n\}\big\}$ defined by $M\mapsto\I(M)$. Then $\I(\cdot)$ is a bijection.

\end{lem}
\begin{proof}
The stable matching $M_n^*$ is associated to $\I(M_n^*)=\emptyset$.
The mapping $\I(\cdot)$ is injective since if $M,M'\in\M_n$ and $M\neq M'$ then $\I(M)\neq\I(M')$.
To see that $\I(\cdot)$ is surjective, fix $K\le n-1$ and $A=\{i_0,i_1,\ldots,i_K\}\in\s$ such that $1\le i_0<i_1<\ldots<i_K\le n$. The matching $M\in\M_n$ such that $\I(M)=A$ is given by
$$
M=\big\{(\ell_{i_{k+1}},f_{i_k}):0\le k\le K-1\big\}\cup\big\{(\ell_k,f_k):k\notin A\big\}\in\M_n.
$$
\end{proof}

\begin{rem}
Lemma~\ref{lem:Mprop} and Lemma~\ref{lem:M121} imply that every matching $M\in\M_n$, $\I(M)=\{i_0,\ldots,i_K\}$, is associated to a unique deficit-decreasing path in $M\oplus M_n^*$ relative to $M$, given by
$$
P(M)=\ell_{i_0},f_{i_0},\ell_{i_1},f_{i_1},\ldots,\ell_{i_K},f_{i_K}.
$$
Too see this, observe that $M\backslash M_n^*$ is given by the non-horizontal edges in $M$, while $M_n^*\backslash M$ is given by the horizontal edges that are not in $M$. Therefore, by Lemma~\ref{lem:Mprop},
\begin{align*}
M\backslash M_n^* & = \big\{(\ell_{i_1},f_{i_0}),(\ell_{i_2},f_{i_1}),\ldots,(\ell_{i_K},f_{i_{K-1}})\big\},\\
M_n^*\backslash M & = \big\{(\ell_{i_0},f_{i_0}),(\ell_{i_1},f_{i_1}),\ldots,(\ell_{i_{K}},f_{i_{K}})\big\},
\end{align*}
and the set of edges in $P(M)$ is equal to $M\oplus M_n^*$.
The uniqueness of $P(M)$ follows since $\I(M)$ is unique by Lemma~\ref{lem:M121} and there is no other way to connect the poor leader $\ell_{i_0}$ and the unmatched follower $f_{i_K}$ with a path.
This suggests that, given a matching $M\in\M_n$, the unique deficit-decreasing path $P(M)$ must be ``solved'' in order to reach the stable matching of $G_n$.
\end{rem}

\subsection{The tree $T^*_m$}
\begin{defn}
Let $T_1$ be a labeled
rooted tree with a singleton node with label $1$.
Inductively, for $i\leq 2$, let $T_i$ be the labeled
rooted tree whose root is labeled with
$i$ and its $i-1$ children are the roots of copies of $T_1,\ldots,T_{i-1}$. We define $T^*_m$ to be the tree with a root with label $m+1$ whose only child
is the root of a copy of $T_m$ (see Figure~\ref{fig:tree} for a visual representation).
Let $r^*$ denote the root of $T^*_m$.
\end{defn}
\begin{figure}[h!]
\centering
% \usebox{\Jn}
\includegraphics[scale=0.4]{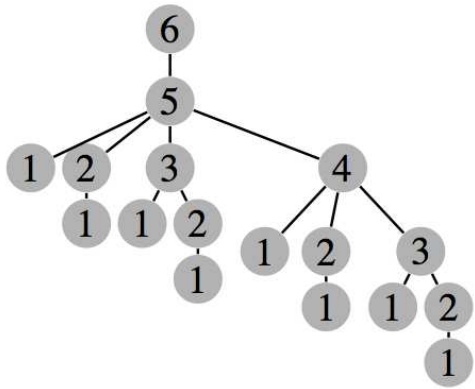}
\caption{The three $T^*_m$ for $m=5$.}
\label{fig:tree}
\end{figure}

We show that the hitting time of $r^*$ for a random walk on $T^*_m$ starting at any node $u\neq r^*$ is exponential in $m$ with high probability.
For a node $u\neq r^*$, we call the edge that connects $u$ to its parent
 $u$'s {\em exit} edge.
For any subtree $T_i \subset T^*_m$, let $Z_i$ be the random variable denoting the number of steps that
it takes for a walk starting at the root of $T_i$ to exit $T_i$ (that is, to hit the parent of the root of $T_i$).
The following lemma provides an exponential lower bound on $Z_i$.
\begin{lem}\label{lem:technical}
There exist positive constants $\alpha,\gamma>0$ such that, for all $i\ge 2$,
\[ \Pr[Z_i \ge \gamma\cdot 2^{i/(\alpha\log^2 i)}] \ge 1- \frac{1}{\log i}.\]
\end{lem}
\begin{proof}
\noindent
We proceed by induction on $i$. For convenience, define $g(i)=\alpha\log^2{i}$ and $f(i) = \gamma\cdot 2^{i/g(i)}$ for some $\alpha,\gamma>0$. For any $\alpha>0$ and $i\ge 2$, we can choose $\gamma>0$ such that $f(i)\le 1$; therefore, as $Z_i\ge 1$ with probability $1$, the claim holds trivially for any $i\le i^*$, where $i^*$ is a
 suitably large constant.

Now consider any $i\ge i^*$ and suppose the claim holds up to $i-1$. Every time the walk is on the root of $T_i$, it exits $T_i$ with probability $1/i$ (since the root of $T_i$ has $i$ neighbors: one parent and $i-1$ children). Therefore, letting $E_t\
$ be the event that the first $t$ times the walk is on the root of $T_i$ it does \emph{not} exit $T_i$, we have $\Pr[E_t]\ge 1-t/i$. Let $t=i/(2\log{i})$, and let $D_j$, $1\le j\le t$, be the event that, when it is on the root of $T_i$ for the $j$-th time,
the walk moves to the root of one of the subtrees $T_{i-g(i)},\ldots, T_{i-1}$ \emph{and} takes at least $f(i-g(i))$ steps to exit that subtree. For $1\le j\le t$, we have
\begin{align*}
\Pr[D_j ~|~ E_t] &\ge \frac{g(i)}{i}\cdot \Pr[Z_{i-g(i)} \ge f(i-g(i)) ]\\
&\ge \frac{g(i)}{i}\cdot\left(1-\frac{1}{\log(i-g(i))}\right),
\end{align*}
by the induction hypothesis on $Z_{i-g(i)}$. Letting $\chi_j$ be the indicator function of the event $D_j$ for $1\le j\le t$, the probability that at least two of the events
$D_j$ happen, given $E_t$, is lower bounded by:
\begin{align*}
\Pr\left[\sum_{j=1}^t \chi_j \ge 2 \ \Bigg|\ E_t\right] & \ge \Pr\left[ \sum_{j=1}^{t/2}\chi_j \ge 1, \sum_{j=t/2+1}^{t}\chi_j \ge 1 \ \Bigg|\ E_t\right]\\
& = \Pr\left[ \sum_{j=1}^{t/2}\chi_j \ge 1 \ \Bigg|\ E_t\right]^2.
\end{align*}
By union bound, we can write
\begin{align*}
&\Pr\left[ \sum_{j=1}^{t/2}\chi_j \ge 1 \ \Bigg|\ E_t\right]\\
&\quad \ge 1- \prod_{i=1}^{t/2}\left(1-\Pr[D_j| E_t]\right)\\
&\quad \ge 1- \left(1-\frac{g(i)}{i}\left(1-\frac{1}{\log(i-g(i))}\right)\right)^{t/2}\\
&\quad \ge 1- \exp{\left[-\frac{\alpha\log{i}}{4}\left(1-\frac{1}{\log(i-g(i))}\right)\right] } \ge 1-\frac{1}{i^{\alpha/8}},
\end{align*}
where the last step holds for $i$ sufficiently large so that $\log(i-g(i))\ge 2$. This implies that
\[ \Pr\left[\sum_{j=1}^t \chi_j \ge 2 \ \Bigg|\ E_t\right] \ge \left(1-\frac{1}{i^{\alpha/8}}\right)^2 \ge 1-\frac{2}{i^{\alpha/8}}. \]
Therefore, we conclude that
\begin{align*}
\Pr[Z_i\ge 2\cdot f(i-g(i))] &\ge \Pr\left[\sum_{j=1}^t \chi_j \ge 2\right]\\
&\ge \Pr\left[\sum_{j=1}^t \chi_j \ge 2 \ \Bigg|\ E_t\right]\Pr[E_t]\\
&\ge \left(1-\frac{2}{i^{\alpha/8}}\right)\left(1-\frac{t}{i}\right)
\ge 1-\frac{1}{\log{i}},
\end{align*}
where the last step holds by choosing $\alpha$ sufficiently large. The claim follows since $2\cdot f(i-g(i))\ge f(i)$.
\end{proof}

Note that a random walk starting at any node $u\neq r^*$ has to exit $T_m$ before hitting $r^*$. Therefore, an application of Lemma~\ref{lem:technical} to $T_m$ yields a
lower bound to the hitting time of $r^*$ when starting at any node $u\neq r^*$.

\begin{cor}\label{cor:hitting}
The hitting time of $r^*$ of a random walk starting at any node $u\neq r^*$ is $2^{\Omega(n/\log^2{n})}$ with high probability.
\end{cor}

\subsection{The dynamics of the algorithm starting from $M\in\M_n$}
For ease of presentation, we set the probability parameters of the algorithms to $p=q=1$. Setting $p=1$ means that a poor leader always proposes to a follower. Setting $q=1$ means that a follower always accepts an incoming request. Our result holds for any choice of $p$ and $q$.

By Lemma \ref{lem:M121}, every matching $M\in\M_n\cup\{M^*_n\}$ is uniquely identified by the set $\I(M)=\{k:(\ell_k,f_k)\notin M\}$.

\begin{defn}[The height of a matching]
Let $M\in\M_n$, $\I(M)=\{i_0,\ldots,i_K\}$. The height $h(M)$ of $M$ is defined as follows. If $K=0$ then $h(M)=0$. If $K\ge 1$ then $h(M)=i_{K-1}\in\{1\,\ldots,n-1\}$.
\end{defn}

For a matching $M\in\M_n$ such that $h(M)>0$ we can write $\I(M)=\{i_0,\ldots,h(M),i_K\}$.
For each $t\ge 0$, let $M(t)$ be the matching at the beginning of round $t$ of the algorithm, and for ease of notation let $\I(t)=\I(M(t))$. For a matching $M\in\M_n$ let
$$\tau^*(M)=\min\big\{t:M(t)=M^*_n\vert M(0)=M\big\}$$
be the number of steps that the algorithm needs to reach the stable matching starting from $M$.

Note that, with $p=q=1$, $t^*(M)=1$ for every $M\in\M_n$ such that $h(M)=0$ (that is, $\vert\I(M)\vert=1$), since according to the algorithm leaders prefer unmatched followers. We are interested in relating $\tau^*(m)$ and $h(M)$ for every matching $M\in\M_n$ such that $h(M)>0$ (that is, $\vert\I(M)\vert>1$).

We study how the matching evolves over time through the Markov process $\{\I(t):0\le t\le \tau^*(M)\}$. Since $\I(M^*_n)=\emptyset$, $\tau^*(M)=\min\{t:\I(t)=\emptyset\}$. The state space of the Markov process is given by the set $\s$ defined in Lemma \ref{lem:M121}. The transition probabilities are characterized by the following lemma.

\begin{lem}
\label{lem:transProb}
Conditional on $\I(t)=I\in\s$, $\vert I\vert>1$, the transition probabilities at time $t$ are given by
\begin{align*}
\Pr\Big(\I(t+1)=I'\big|\I(t)=I\Big)=\frac{1}{\min I}
\qquad\text{if}\quad
I'\in \Big\{ I\cup\{k\}:k<\min I \Big\}\cup \Big\{ I\backslash\{\min I\} \Big\},
\end{align*}
and $0$ otherwise. Moreover $\Pr(\I(t+1)=\emptyset\vert\I(t)=\emptyset)=1$, and $\Pr(\I(t+1)=\emptyset\vert\I(t)=I)=1$ for every $I$ sich that $\vert I\vert=1$.
\end{lem}
\begin{proof}
The case of $\I(t)=\emptyset$ corresponds to the stable matching $M^*_n$, which is an absorbing state for the Markov process.
In the case of $\vert \I(t)\vert=1$, we have that $h(M)=0$, and $p=q=1$ implies that  that $\I(t+1)=\emptyset$.

Consider now $\vert I\vert>1$. Conditional on $\I(t)=I$, the poor leader is $\ell_{\min I}$ and has degree $\min I$ and neighborhood $N_{\min I}=\{f_1,\ldots,f_{\min I}\}$, and chooses one of the followers in $N_{\min I}$ uniformly at random. If $\ell_{\min I}$ chooses follower $f_k$ for some $k<\min I$ then the leader $\ell_k$ becomes poor, since by property (3) of Lemma~\ref{lem:Mprop} $\ell_k$ was matched to $f_k$ in $M(t)$, and we have that $\I(t+1)=I\cup\{k\}$.
If instead $\ell_{\min I}$ chooses follower $f_{\min I}$ (matched to $\ell_{\min(I\backslash \min I)}$ in $M(t)$ by property (4) of Lemma~\ref{lem:Mprop}), then $\I(t+1)=I\backslash\{\min I\}$.
\end{proof}

For every matching $M\in\M_n$ such that $h(M)>0$ and $\I(M)=\{i_0,\ldots,i_K\}$, define the matching $\LL(M)=\{(\ell_j,f_j):j\neq i_K\}$ and $\tau(M)=\min\{t:M(t)=\LL(M)\}$, and observe that $h(\LL(M))=0$ and $\tau^*(M)>\tau(M)$ (in particular, $\tau^*(M)=1+\tau(M)$ for $p=q=1$).

For every matching $M$ such that $\vert\I(M)\vert>1$, let $\R(M)$ be the set of the matchings in $\M_n$ that can be reached from $M$ (after one or multiple steps).
According to the transition probabilities defined by Lemma~\ref{lem:transProb}, it is easy to see that
\begin{align*}
\R(M)= \Big\{\LL(M)\Big\}\cup
\Big\{M'\in\M_n:I(M')=A\cup\{h(M),i_{K}\},A\subseteq\{1,\ldots,h(M)-1\}\Big\}.
\end{align*}
Observe that every $M'\in\R(M)\backslash\{\LL(M)\}$ has height $h(M')=h(M)$. The following lemma characterizes the one-to-one correspondence between matchings in  $\R(M)$ and nodes of the tree $T^*_{h(M)}$. 
\begin{lem}
Consider the mapping $\omega(\cdot)$ from $\R(M)$ to $T^*_{h(M)}$ defined as follows.
Let $\omega(\LL(M))=r$, where $r$ is the root of $T^*_{h(M)}$.
For $M'\in\R(M)\backslash\{\LL(M)\}$ and $\I(M')=I$, let $\omega(M')$ be the node of $T^*_{h(M)}$ with label $\min I$ and connected to the root with a path of nodes labeled by the sorted indexes in $I\backslash\{\min I\}$.
Then $\omega(\cdot)$ is a bijection.
\end{lem}
The proof is omitted since it directly follows from the construction of the tree $T^*_{h(M)}$ and the mapping $\I(\cdot)$.

\begin{lem}
The stochastic process $\{\I(t):0\le t\le\tau(M)\vert M(0)=M\}$ is equivalent to a random walk on $T^*_{h(M)}$ starting at $\omega(M)$.
\end{lem}
\begin{proof}
It suffices to show that the transition probabilities between two matchings $M_1,M_2\in\R(M)$ are nonzero if and only if the nodes $\omega(M_1)$ and $\omega(M_2)$ are adjacent in $T^*_{h(M)}$.
To prove the ``only if'' direction, assume that $M_1,M_2\in\R(M)$ are such that there is a nonzero transition probability from $M_1$ to $M_2$ (and therefore from $M_2$ to $M_1$). Let $\I(M_1)=I_1$ and $\I(M_2)=I_2$. According to the transition probabilities given above, there are two possible cases. In the first case, $I_2=I_1\cup\{k\}$ for some $k<\min I_1$, and $\omega(M_1)$ is a child of $\omega(M_1)$. In the second case $I_2=I_1\backslash \{\min I_1\}$ and $\omega(M_2)$ is the parent of $\omega(M_1)$. The proof of the other direction is similar.
\end{proof}
To summarize, the number of steps that the algorithm needs to reach the stable matching of $G_n$ starting from $M\in\M_n$ with $h(M)>0$ is upper bounded by the time $\tau(M)$ to reach the matching $\LL(M)$, and reaching $\LL(M)$  is equivalent to reaching the root of $T^*_{h(M)}$ starting from the node $\omega(M)$.
By Corollary~\ref{cor:hitting}, $\tau(M)$ is exponentially large in $h(M)$ with high probability. To complete the proof of the theorem, we show that, for any constant $0<\gamma<1$, a $1-O(n2^{-(1-\gamma)n})$ fraction of the matchings $M\in\M_n$ have $h(M)\ge \gamma n$. This is done through a counting argument.

\subsection{The fraction of the matchings  $M\in\M_n$ such that $h(M)\ge\gamma n$}

Let $N$ be the number of matchings in $\M_n$. Fixed a constant $0<\gamma<1$, let $\M_{\gamma}=\{M\in\M_n:h(M)<\gamma n\}$ and let $N_{\gamma}=\vert\M_{\gamma}\vert$.
For $j=0,\ldots,n-1$, let $N(j)$ be the number of matchings $M\in\M_n$ such that $h(M)=j$.
It follows that
$$
N  =\sum_{j=0}^{n-1}N(j),\qquad
N_{\gamma} \le\sum_{j=0}^{\lceil\gamma n\rceil-1}N(j).
$$
\begin{lem}
$N(0)=n$ and 
$N(j)=(n-j)2^{j-1}$ for all $j=1,\ldots,n-1$.
\end{lem}
\begin{proof}
$N(0)=n$ since there are $n$ matchings $M$ with $h(M)=0$, that is, the matchings $\{(\ell_j,f_j):j\neq k\}$ for $1\le k\le n$.

Fix $j\in\{1,\ldots,n-1\}$. By Lemma~\ref{lem:M121}, a matching $M\in\M_n$ with $h(M)=j$ is uniquely identified by a set
$\I(M)=\{i_0,\ldots,i_{K-1},i_K\}$ for some $1\le K\le n-1$ and $i_{K-1}=j$. Since $\I(\cdot)$ is a bijection, to determine $N(j)$ we need to count all subsets of $\{1\ldots,n\}$ of the form $\{i_0,\ldots,j,i_K\}$. There are $2^{j-1}$ subsets of $\{1,\ldots,j-1\}$ and $n-j$ ways to choose $i_K\in\{j+1,\ldots,n\}$, thus $N(j)=(n-j)2^{j-1}$.

\end{proof}

We now show that for any constant $0<\gamma<1$, the fraction of matchings $M\in\M_n$ such that $h(M)<\gamma n$ goes to zero exponentially fast in $n$
\begin{lem}
Fix $0<\gamma<1$. Then, $N_{\gamma}/N=O(n2^{-(1-\gamma)n})$.
\end{lem}
\begin{proof}
We first compute $N$.
$$N=\sum_{i=0}^{n-1}N(i)=n+\sum_{i=1}^{n-1}(n-i)2^{i-1}=n+n\sum_{i=0}^{n-2}2^i-\sum_{i=1}^{n-1}i2^{i-1}.$$
The second sum can be shown (e.g. by induction) to be equal to $(n-1)+(n-2)(2^{n-1}-1)$. Therefore,
$$N=n+n(2^{n-1}-1)-(n-1)-(n-2)(2^{n-1}-1)=2^n-1=\Omega(2^n).$$

Similarly, letting $k=\lceil \gamma n\rceil$ we have that,
\begin{align*}
N_{\gamma}&\le \sum_{i=0}^{k -1}N(i)=n+n\sum_{i=0}^{k-2}2^i-\sum_{i=1}^{k-1}i2^{i-1}\\
&=n+n(2^{k-1}-1)-(k-1)-(k-2)(2^{k-1}-1)\\
&=2^{k-1}(n-k-2)-1=O(n2^{\lceil \gamma n\rceil}).
\end{align*}
Therefore, the fraction of matchings in $\M_n$ with height $h(M)<\gamma n$ is $N_{\gamma}/N=O(n2^{-(1-\gamma) n})$.
\end{proof}

\bibliographystyle{plain}
\bibliography{CDC_2012.bib}

\end{document}